\documentclass{article} 
\usepackage[american]{babel}

\usepackage{mathtools} 
\usepackage{amsfonts}
\usepackage{amsthm}
\usepackage{amssymb}
\usepackage{authblk}
\usepackage{float}
\usepackage[caption = false]{subfig}
\usepackage{comment}
\usepackage{mleftright}

\newtheorem{theorem}{Theorem}[section]
\newtheorem{corollary}{Corollary}[theorem]
\newtheorem{lemma}[theorem]{Lemma}   
\newtheorem{assumption}{Assumption}[section]

\title{Decentralized Langevin Dynamics over a Directed Graph}
\author{Alexander Kolesov\footnote{Moscow Institute of Physics and Technology kolesov.as@phystech.edu} and Vyacheslav Kungurtsev \footnote{Czech Technical University in Prague vyacheslav.kungurtsev@fel.cvut.cz}}
\begin{document}
\maketitle

\begin{abstract}
    The prevalence of technologies in the space of the Internet of Things and use of multi-processing computing platforms to aid in the computation required to perform learning and inference from large volumes of data has necessitated the extensive study of algorithms on \emph{decentralized} platforms. In these settings, computing nodes send and receive data across graph-structured communication links, and using a combination of local computation and consensus-seeking communication, cooperately solve a problem of interest. Recently, Langevin dynamics as a tool for high dimensional sampling and posterior Bayesian inference has been studied in the context of a decentralized operation. However, this work has been limited to undirected graphs, wherein all communication is two-sided, i.e., if node A can send data to node B, then node B can also send data to node A. We extend the state of the art in considering Langevin dynamics on directed graphs. 
\end{abstract}

\section{Introduction}
Recently, there has been a surge in the interest of using computing platforms
situated on a network, modeled as a graph with vertices and edges. Originally
termed \emph{distributed}, this has more recently transformed as being denoted
as \emph{decentralized}, to contrast with data-parallel distributed computation
methods in high performance computing. In this setting, a number of agents,
defined as vertices $\mathcal{V}$ with communication links $\mathcal{E}$ in
a fully connected graph $\mathcal{G}=(\mathcal{V},\mathcal{E})$ perform alternating sequences of steps of local computation and consensus-based communication to cooperatively solve some problem in optimization, learning, or inference. This line of research began with the seminal work on optimization in~\cite{nedic2009distributed} although precursors in network control theory exist. 

The rise of distributed algorithms can come from an underlying physical reality
of a problem incorporating information that is distributed across a network of agents that they must cooperatively solve, or
a situation arising from the contemporary age of ``big data" in statistics and machine learning. In this case one computer is unable to store the entire dataset for learning and there is a practical necessity for separating it across a set of machines. As soon as we divide data to many computers, then the connectivity structure of such computers will constitute the distributed network. In case of the absence of
a central master/server machine the network is described as decentralized. 

On undirected graphs, each communication link, modeled as an edge
in the graph $e=(v_1,v_2)\in\mathcal{E}$ for $v_1,v_2\in\mathcal{V}$ is such that communication is bi-directional, i.e., in this case $(v_2,v_1)\in\mathcal{E}$ as well. In a \emph{directed} graph, this may not in general be the case, i.e., it could be that node $v_1$ can send data to note $v_2$, but not vice versa. In general, connections can be \emph{time-varying}, meaning that nodes connect and drop out at various moments in time. This issue
of communication crashes is one of practical importance as it is important
that the optimization, inference or learning procedure be robust with respect to such occurrences.

As an illustration, let us consider the multi-agent system that contains 4 nodes depicted in Figure~\ref{fig:directednetwork}. Let the connection between the nodes at a certain time $t_{1}$ be as in Figure ~\ref{fig:directednetwork}(a). Subsequently, for some reason the connection between from first to the second agent fails, then the graph loses the directed edge between them at a certain time $t_{2}$ that is depicted in figure ~\ref{fig:directednetwork}(b). Hence, the network is time-varying directed graph.

\begin{figure}
\centering
\subfloat[at time $t_{1}$]{\includegraphics[width = 1.68in]{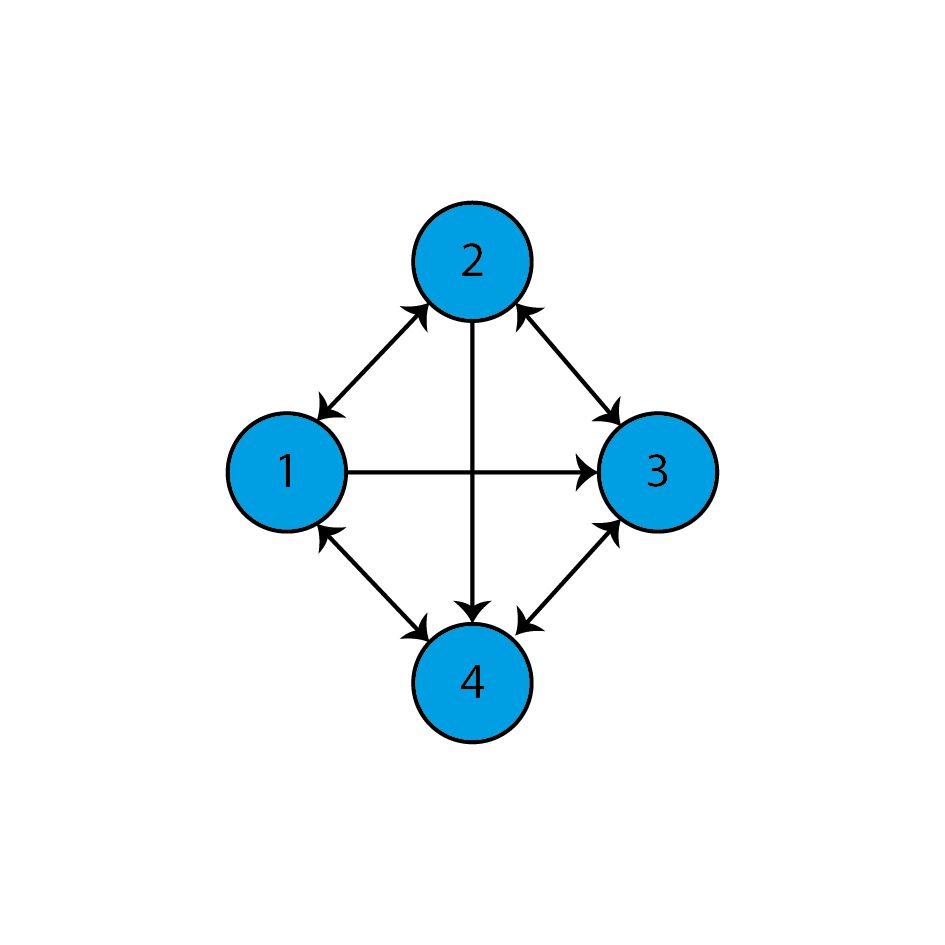}} 
\subfloat[at time $t_{2}$]{\includegraphics[width = 1.68in]{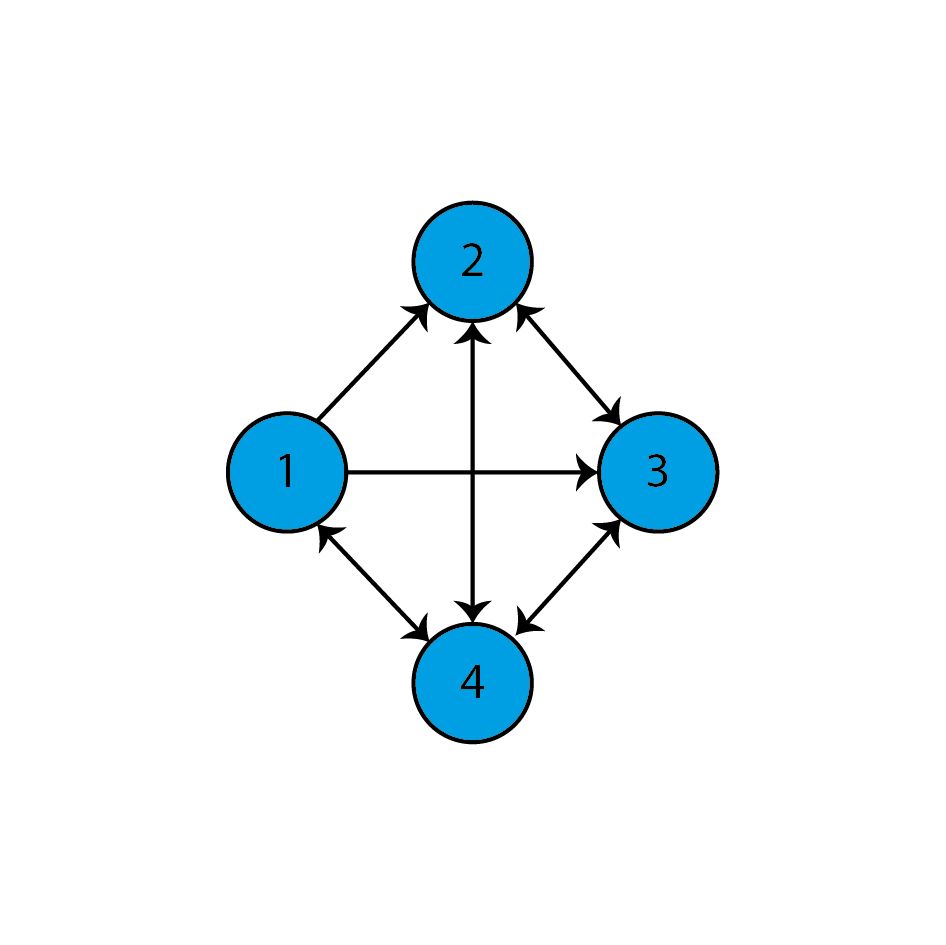}}
\caption{Time Varying Directed Network}
\label{fig:directednetwork}
\end{figure}

Having realised the structure of a decentralized multi-agent system, one can move on the problem statement. The task is to sample from a probability distribution that is known up to a normalizing constant. For instance, such a problem arises in Bayesian Inference, when we want to calculate a posterior distribution $p(x | \theta) $  for a set of parameters $x$, knowing a prior knowledge $p(x)$ and a  likelihood $p(\theta|x)$, where $\theta$ is data. In accordance with the Bayes' theorem ~\eqref{eq:Bayess formula}
\begin{equation}
    \centering
    \label{eq:Bayess formula}
    p(x|\theta) = \frac{p(\theta|x)p(x)}{\int_{X} p(\theta|x)p(x)dx},
\end{equation}
there is a constant in a denominator that constitutes the multiplication of the prior and the likelihood, integrated across the parameters. This integral normalizing constant
is typically difficult if not impossible to compute, and hence sampling methods often seek to avoid doing so.

Markov Chain Monte Carlo (MCMC) methods are able to obtain diverse samples from a target posterior distribution. The core of the techniques is to sample a proposal distribution which is then accepted or rejected. There are many variants of these proposal distributions. More recently, it was found that, for especially 
log concave potentials in general, discretizations of physically motivated stochastic differential equations are able to efficiently for distributions with a high dimension for the parameter $x$.
One of these proposals is based on Hamiltonian dynamics, which was observed at first in ~\cite{Neal2012}, while Langevin dynamics underlies another type of proposal, introduced for posterior Bayesian inference
with the seminal work~\cite{welling2011bayesian}. The overdamped Langevin diffusion equation, in particular, arrives at a stationary distribution characterized by a pdf of the form $e^{-U(x)}/Z$, where $Z$ is the normalizing
constant, when the diffusion drift term is $\nabla U(x)$. This form of potential is common for Gibbs type distributions and likelihood functions arising from exponential families.
Formally, we consider finding the stationary distribution of a pdf wherein there exists a potential given by,
\[
U(x) = \sum\limits_{i=1}^m U_i(x)
\]
where for presentation we subsume the dependence on the data $\theta$ within $U_i(\cdot)$, i.e., $U_i(x) = p(\theta_i|x)p(x)$, wherein we drop the normalization constant, and
where node $i$ knows only the strongly convex function $U_{i}$: $\mathbb{R}^{d} \rightarrow \mathbb{R}$. 

We make the following assumption on $\{U_i(\cdot)\}$.
\begin{assumption}\label{as:potentialas}
Each $U_i(x)$ is Lipschitz continuously differentiable with constant $L$, and
strongly convex with constant $\mu$. Furthermore $\sum\limits_{i=1}^m U_i(x)$
is also Lipschitz continuously differentiable and strongly convex
with, without loss of generality, the same constants.
\end{assumption}

 
\section{Background}
\subsection{The network structure}
As described above, the multi-agent system (network) constitutes directed and time-varying connections between nodes with changing in- and out-neighbours. The communication structure is modeled as a graph denoted by ${\mathcal{G}}(\mathcal{V},\mathcal{E})$. That is defined as a set of vertex $\mathcal{V}$ = \{1,...,m\} and we will use $\mathcal{E}(t)$
to label a set of edges at a certain time \textit{t} throughout the article. 

Let us make some assumptions on the network. It was shown in the related work ~\cite{NedicDirected2014}, that the property of B-strongly-connectedness is sufficient to derive different bounds on the speed of information propagation.
So we require that the sequence $\{\mathcal{G}(t)\}$ is $B$-strongly-connected. In other words, there exists  positive integer $B$ such that the graph's edge set is strongly connected for any non-negative $k$. Formally:
\[ \mathcal{E}_{B}(k) = \bigcup^{(k + 1)B - 1}_{i = kB}  \mathcal{E}(i) \]
and the graph $\mathcal{G}$ is connected.
Since the observed multi-agent system is directed and time-varying, then one has to introduce in- and out-neighboors for each node i at the current time \textit{t}. 
\[ \mathcal{N}_{i}^{in}(t) = \{j | (j,i) \in \mathcal{E}(t)\} \cup \{i\} \]

\[ \mathcal{N}_{i}^{out}(t) =\{j | (i,j) \in \mathcal{E}(t)\} \cup \{i\} \]

The authors of the article ~\cite{NedicDirected2014} note that the subgradient-push algorithm they developed for decentralized optimization requires only  knowledge of out-degree for each node i. So, one should define the out-degree of node i at time \textit{t} as:
\begin{equation}\label{eq:dout}
    d_{i}(t) = |\mathcal{N}^{out}_{i}(t)|
\end{equation}   

Now, let us introduce the mixing matrix $A(t)$. As an illustration, consider the graph that is depicted in Figure ~\ref{Circular graph}.

\begin{figure}
\centering
\subfloat[The graph]{\includegraphics[width = 1.8in]{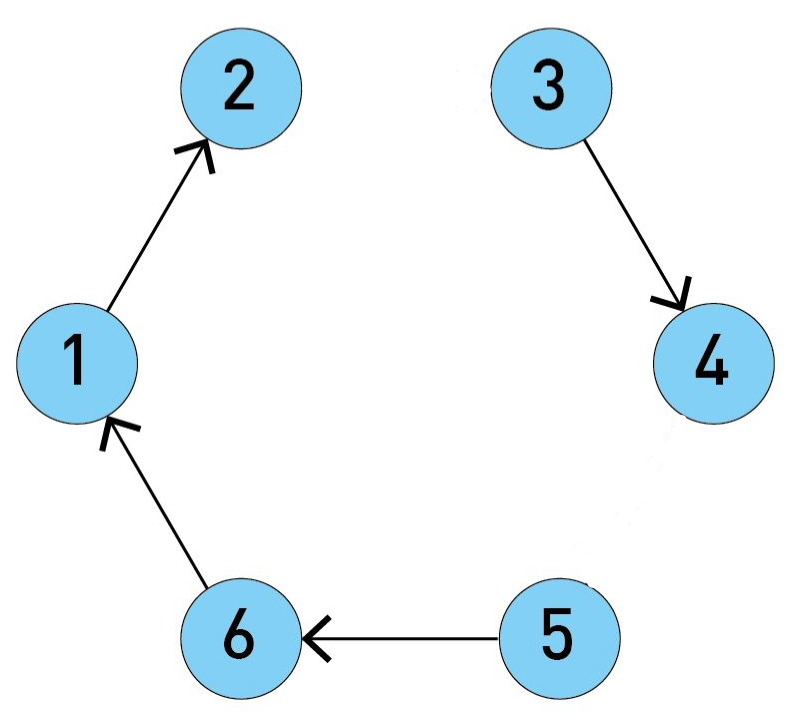}}
\caption{Almost circular graph}
\label{Circular graph}
\end{figure}

This graph is composed of 6 nodes and at this moment it is not fully connected (but is expected to complete its links over time).
Taking the first agent, one can see, that it sends his signal to the second agent and itself. Let us then set $a_{11}$ equal to $\frac{1}{2}$ and  $a_{21}$ the same. Continuing, the mixing matrix for the graph is written in ~\eqref{eq:mixing matrix}. One can notice, that the sum of elements over any column is equal to 1, so the matrix $A(t)$ is a column-stochastic. 

\begin{equation}
A(t) =  
\begin{bmatrix}
 
\frac{1}{2} & 0 & 0           & 0 & 0 & \frac{1}{2}\\
\frac{1}{2} & 1 & 0           & 0 & 0 & 0\\
0           & 0 & \frac{1}{2} & 0 & 0 & 0\\
0           & 0 & \frac{1}{2} & 1 & 0 & 0\\
0           & 0 & 0           & 0 & 1 & 0\\
0           & 0 & 0           & 0 & 0 & \frac{1}{2}
\end{bmatrix}
\label{eq:mixing matrix}
\end{equation}

We make this assumption formally,
\begin{assumption}
The mixing matrix $A(t)$ is column-stochastic for all $t$.
\end{assumption}

\subsection {Elements of spectral theory}
Let us present some basic background on the spectral theory of mixing matrices. Introduce the following constants $\lambda$ and $\delta$ associated to a graph sequence $\mathcal{G}(t)$. The value of $\delta$ corresponds to, once you continuously multiply the matrices $A(t)$, the smallest value of an entry of that product matrix. 
There is always a lower bound and in accordance with assumption V.1 in ~\cite{Rogozin}, the standard condition is,
\[ \delta \geq \frac{1}{m^{mB}}\]
One can interpret this as a lower bound on the weight any node assigns to past information from another node.
 
As for $\lambda$, for the case of fixed graphs, the variable $\lambda$ will represent the connectivity of the graph, and it is usually the second eigenvalue of the graph adjacency matrix. Usually, the better connected the graph is, the bigger $\lambda$ is, for poorly connected graphs $\lambda$ is almost zero. However, when the graph changes with time, there is no direct interpretation via the eigenvalue, because every graph in the sequence will have a different eigenvalue.Thus one can interpret $\lambda$ in the time-varying case as a lower bound for the connectivity of the time-varying graph.
\begin{equation}\label{eq:lamb}
\lambda \leq \left(1 - \frac{1}{m^{mB}}\right)^{\frac{1}{mB}}
\end{equation}



\subsection{Previous Work}
Stochastic Gradient Langevin Dynamics was introduced for posterior Bayesian inference
with the seminal work~\cite{welling2011bayesian}. 

The paper~\cite{kungurtsev2020stochastic} presented the first application of applying Langevin dynamics, and thus finding a stationary distribution, in the context of a decentralized setting of data distributed across a network. The work~\cite{gurbuzbalaban2020decentralized} extended the framework to consider the momentum based Hamiltonian Langevin dynamics. Finally~\cite{parayil2020decentralized} considered the problem relaxing the conditions for convexity of the potential (log-concavity of the distribution function) to a log-Sobolev inequality, thus permitting posterior inference for a wider class of decentralized problems.

In regards to decentralized optimization, besides the seminal work~\cite{NedicDirected2014}, there is the paper considering stochastic
gradients in~\cite{nedic2016stochastic}. For a nice recent survey on 
decentralized optimization in machine learning see~\cite{Nedic_ML}.

\subsection{Probability Distances}

In order to present our theoretical results, we must introduce some notation regarding the computation of the distance between two probability distributions. 
A transference plan $\zeta(\mu,\nu)$
of two probability measures $\mu$ and $\nu$ on $\mathcal{B}(\mathbb{R}^d)$ is itself a probability
measure on $(\mathbb{R}^d\times\mathbb{R}^d,\mathcal{B}(\mathbb{R}^d\times\mathbb{R}^d))$ that satisfies: for all measurable $A\subseteq\mathbb{R}^d$,
it holds that $\zeta(A\times\mathbb{R}^d)=\mu(A)$ and
$\zeta(\mathbb{R}^d\times A)=\nu(A)$. We denote by $\Pi(\mu,\nu)$ the set
of transference plans of $\mu$ and $\nu$.Two  
$\mathbb{R}^d$-valued random variables $(X,Y)$ is a \emph{coupling} of $\mu$
and $\nu$ if there exists a $\zeta\in\Pi(\mu,\nu)$ such that $(X,Y)$
are distributed according to $\zeta$. 
The Wasserstein distance of order two is,
\[
W_2(\mu,\nu) = \left(\inf_{\zeta\in\Pi(\mu,\nu)} 
\int_{\mathbb{R}^d\times\mathbb{R}^d} \|x-y\|^2 d\zeta(x,y)\right)^{1/2}.
\]
For all $\mu,\nu$ there exists a $\zeta^*\in\Pi(\mu,\nu)$ realizing
the $\inf$, i.e., for any coupling $(X,Y)$ distributed according
to $\zeta^*$ we have $W_2(\mu,\nu)=\mathbb{E}[\|X-Y\|^2]^{1/2}$, defined as
the \emph{optimal transference plan} and \emph{optimal coupling} associated with $W_2$.
The space $\mathcal{P}_2(\mathbb{R}^d)$ is the set of finite second
moment probability measures and together with $W_2$ is a complete separable metric
space.

We denote $\mu \ll \nu$ to mean that $\mu$ is absolutely continuous w.r.t. $\nu$. Recall the Kullback-Leibler (KL) divergence
of $\mu$ from $\nu$ is defined by,
\[
KL(\mu|\nu) = \left\{ \begin{array}{ll} \int_{\mathbb{R}^d} \frac{d\mu}{d\nu}(x)\log\left(
\frac{d\mu}{d\nu}(x)\right), & \text{if }\mu\ll \nu,\\
\infty & \text{otherwise.} \end{array}\right.
\]

\section{Algorithm}\label{s:alg}

Now, we propose an algorithm to solve the problem of sampling from the target distribution in a decentralized setting with directed graphs. As already mentioned, the perturbed push-sum protocol, which was published in the work ~\cite{NedicDirected2014} is at the heart of our developed algorithm. The aspect of the procedure seeking consensus is actually exactly the same. One should take the mixing matrix, which is set by the current graph's structure, and multiply this matrix $A(t)$ by the vector of coordinates $X$ which is a stack comrpomising of concatenated vectors $\{X_{(i)}\}$, where each $X_{(i)}$
is in turn a stack of components of coordinate on $i$'th node.
As soon as we compute the consensus step, then we can compute the balancing vector $Y_{(i)}$ for each agent. The (stochastic) gradient is evaluated at another vector $Z_{(i)}$ that is a balanced weighted average of $X_{(i)}$ and its neighbors, and the update, together with the Gaussian noise, is applied to $X_{(i)}$. 

The set of quantities available and computed by each agent is,
\[
\left\{ W_{(i)}(t),y_{(i)}(t), Z_{(i)}(t), X_{(i)}(t)\right\} 
\]
where aside from the scalar $y_{(i)}(t)$, all other quantities of $d$-dimensional vectors.

We now present the specific Algorithm, as inspired by the merging of the perturbed push sum approach given in~\cite{NedicDirected2014} and the unadjusted Langevin algorithm, as described for instance in~\cite{ULAMouline}.

The Algorithm, from the perspective of agent $i$, is defined as follows. To begin with, we set, for all $i$, the quantities,
\[
\left\{ W_{(i)}(0), Z_{(i)}(0), X_{(i)}(0)\right\} 
\]
to zero and $y_{(i)}(0)=1$. Recall that at each iteration $d_i(t)$ is defined by~\eqref{eq:dout}. Each agent $i$ iterates the set of equations~\eqref{eq:langdynamicsout}
\begin{equation}\label{eq:langdynamicsout}
\begin{cases}
\begin{array}{l}
W_{(i)}(t + 1) = \sum\limits_{j\in \mathcal{N}_{i}^{in}(t)} \frac{X_{(j)}(t)}{d_j(t)} \\
Y_{(i)}(t + 1) = \sum\limits_{j\in \mathcal{N}_{i}^{in}(t)} \frac{Y_{(j)}(t)}{d_j(t)} \\ 
Z_{(i)}(t + 1) = \frac{W_{(i)}(t + 1)}{Y_{(i)}(t + 1)} 
\\
X_{(i)}(t + 1) =  \sum\limits_{j\in \mathcal{N}_{i}^{in}(t)} \frac{X_{(j)}(t)}{d_j(t)}\\ \qquad -\alpha(t+1) \nabla U_{(i)}(Z_{(i)}(t + 1))) \\ \qquad  + \sqrt{2 \alpha(t+1)}B_{(i)} (t+1)
\end{array}
\end{cases}
\end{equation}
where we denote $B_{(i)}(t+1) = \sqrt{\alpha(t+1)}\xi_{(i)}(t+1)+R_{(i)}(t+1)$, with $\xi(t+1)$ is a zero bias bounded variance stochastic gradient error and $R_{(i)}(t+1)$ is an isotropic Gaussian random variable. Let us set $\sqrt{\alpha(t+1)}\le 1/\sigma^2:=1/ \mathbb E \|\xi\|$ so that the standard deviation associated with $\sqrt{2\alpha_t}B(t+1)$ is always less than 2.


Let us now consider the iterations on the full stack of vectors $\{X,Y,Z\}$
using the mixing matrix $A$, 
\begin{equation}\label{eq: langmatrices_almost_discret}
\begin{cases}
\begin{array}{l}
Y(t + 1) = A(t)Y(t)\\
Z(t + 1)  = {A(t)X(t)}\oslash{A(t)Y(t)}\\
X(t + 1) = A(t)X(t) - \alpha(t + 1) \nabla \tilde U(Z(t + 1)) \\+ \sqrt{ 2 \alpha(t + 1) }B_{t + 1}
\end{array}
\end{cases}
\end{equation}
where we are defining,
\[
\nabla \tilde U(Z) = \begin{pmatrix} \nabla U_1(Z_{(1)}) \\ 
... \\ \nabla U_i(Z_{(i)}) \\ .,. \\ \nabla U_m(Z_{(m)}) \end{pmatrix}
\]
and we abuse notation with letting $[X\oslash Y]_{(i)}=X_{(i)}/y_{(i)}$

In the theory (and of course the numerical implementation), we shall consider
the discretization, however, out of mathematical interest, we can note that,
considering the form of the Euler-Maruyama (EM) discretization, writing $A(t)X(t) = A(t)X(t) + X(t) - X(t) = (A(t) - I)X(t) + X(t))$, we can notice, out of mathematical interest, that this corresponds to the EM discretization of the dynamics given by (where we now overload $t$ to be continuous),
\[ 
\begin{array}{l}
Z_t = (A_t X_t)\oslash (A_t Y_t) \\ 
dY_t = (A_t - I)Y_t \\
dX_t = (A_t-I)X_t-\alpha_{t}\nabla \tilde U(Z_t)+\sqrt{2\alpha_t} dB_t
\end{array}
\]
where $dB_t$ is a Brownian motion term and clearly the discretization is with a step-size of one. Unlike other works in Langevin dynamics, we do not analyze the SDE directly but only its discretization, however.

\section{Theoretical Results}
We present the sequence of Lemmas and our final convergence result
in measure for the Algorithm defined in Section~\ref{s:alg}. The proofs of the
statements are left to the Supplementary Material.

First we will require a bound in expectation on the gradient vectors,
which we derive by deriving a bound on the expectation of the norm of
the vectors on which they are evaluated, $Z_{(i)}(t)$.
\begin{lemma}\label{lem:helptobound}
    It holds that there exists some compact set $\mathcal{W}$ such that for all $i$
    \[
    \mathbb{E}\left\|\frac{X_{(i)}(t+1)}{Y_{(i)}(t+1)}\right\|\le \left\{\begin{array}{lr} \|Z_{(i)}(t+1)\| &
    \text{if } Z_{(i)}(t+1)\in\mathcal{W} \\
    R & \text{otherwise}\end{array}\right.
    \]
    with $R$ depending on $\mathcal{W}$ and problem constants.
\end{lemma}

\begin{lemma}\label{lem:gradbound}
    It holds that there exists a $C$ such that,
    \[
    \mathbb{E}\|\nabla U_i(Z_{(i)}(t+1))\|\le C
    \]
\end{lemma}

The next statement is the same as Corollary 1 in~\cite{nedic2016stochastic}, where we recall that $\lambda$ is a graph-structure related constant satisfying~\eqref{eq:lamb} which implies that $\lambda<1$,
\begin{lemma}
    It holds that,
    \[
    \begin{array}{l}
    \left\|Z_{(i)}(t+1)-\frac{\sum\limits_{i=1}^m X_{(i)}(t)}{m}\right\|\\ \qquad \le 
    \frac{8}{\delta}\left(\lambda^t\sum\limits_{i=1}^m\|X_{(i)}(0)\|_1+
    \sum\limits_{s=1}^t \lambda^{t-s}\sum\limits_{i=1}^m
    \left\| e_i(s)\right\|\right)
    \end{array}
    \]
    where,
    \[
    e_i(s) = \alpha(s+1) \nabla U_{i}(Z_{(i)}(s + 1)))  + \sqrt{2 \alpha(s+1)}B_{(i)}(s+1)
    \]
\end{lemma}

This together with Lemma~\ref{lem:gradbound} allows us to prove the following
bound on the running sum of the consensus error.
\begin{lemma}\label{lem:conserror}
\[
\begin{array}{l}
    \mathbb{E}\left[\sum\limits_{t=1}^\tau \left\|Z_{(i)}(t+1)-\frac{\sum\limits_{i=1}^m X_{(i)}(t)}{m}\right\| \right] \\
    \qquad \le \frac{8}{\delta}\frac{\lambda}{1-\lambda}\sum\limits_{i=1}^m 
    \|X_{(i)}(0)\|_1+\frac{8}{\delta}\frac{Dm}{1-\lambda} (1+\sqrt{\tau})
\end{array}
\]
\end{lemma}

Now define,
\[
\bar{X}(t) = \frac{\sum\limits_{i=1}^m X_{(i)}(t)}{m}
\]
we have, by the column-stochasticity of $A(t)$,
\begin{equation}\label{eq:avgupdate}
\begin{array}{l}
\bar{X}(t+1)=\bar{X}(t) - \frac{\alpha(t+1)}{m}\sum\limits_{i=1}^m \nabla U(\bar{X}(t)) \\ \qquad +  \frac{\alpha(t+1)}{m}\sum\limits_{i=1}^m \left(\nabla U(\bar{X}(t))-U(Z_{(i)}(t+1))\right)+\bar{B}(t+1)
\end{array}
\end{equation}
Lemma~\ref{lem:conserror} implies the following,
\begin{corollary}\label{cor:cons}
Given any $\gamma>0$, there exists $C_\gamma>0$ such that,
\[
E\|\bar{X}(t))-Z_{(i)}(t+1)\| \le \frac{C_\gamma}{t^{1/2-\gamma}}
\]
\end{corollary}

Let $\bar{\nu}_t$ be the distribution associated with $\bar{X}(t)$

We are now ready to use the arguments in~\cite{dalalyan2019user} regarding perturbed Langevin methods. In particular, we can apply Proposition 2 to state that,
\begin{lemma}
\[
\begin{array}{l}
W_2(\bar{\nu}_{t+1},\pi) \le \rho_{t+1}W_2(\bar{\nu}_t,\pi)+1.65 L (\alpha_{t+1}^3 d)^{1/2}\\ \qquad +\alpha(t+1)\sqrt{d} L \sum\limits_{i=1}^m \mathbb E\|\bar{X}(t)-Z_{(i)}(t+1)\|
\end{array}
\]
where $\rho_{t+1} = \max(1-\mu \alpha(t+1),L\alpha(t+1)-1)$.
\end{lemma}

Now recall
\begin{lemma}\cite[Lemma 2.4]{polyak1987introduction}\label{lem:polyak}
Let $u_k\ge 0$ and,
\[
u_{k+1}\le \left(1-\frac{c}{k}\right) u_k+\frac{d}{k^{p+1}}
\]
with $d>0$, $p>0$ and $c>0$ and $c>p$. Then,
\[
u_k\le d(c-p)^{-1}k^{-p}+o(k^{-p})
\]
\end{lemma}

Apply the previous two Lemmas together with Corollary~\ref{cor:cons} to conclude,
\begin{theorem}\label{th:conv}
Let $\alpha(0)\le \min\left\{\frac{1}{2L},\frac{\mu}{4L^2}\right\}$ and
$\alpha(t)=\alpha(0)/(1+t)$. We have that,
\[
W_2(\bar{\nu}_{k+1},\pi) \le \frac{(C_\gamma m+1.65) Ld^{1/2}}{(\mu\alpha(0)-1/2+\gamma) (1+t)^{(1/2-\gamma)}}+\beta_k
\]
where $\beta_k=o\left((1+t)^{-(1/2-\gamma)}\right)$
\end{theorem}

\section {Numerical Experiments}

\subsection{Notes on Parameter Tuning}
There is a classical practical question in regards to implementing SGLD algorithms for sampling--how do we tune the step-size, considering its dual influence on the convergence as well as the variance of Brownian term? There are some papers, whose main goals are to study the question. The first mentioned related work ~\cite{Neklyudov2018KLdivergence} solves the problem via finding the choice that minimizes the asymptotic KL-divergence. However, it applies only in the case of independent proposal distributions, when new samples are generated regardless of the history. Thus, applying such approach to classic proposals as Random Walk Metropolis (RWM) or Metropolis-Hastings adjusted Langevin algorithm (MALA), it leads to a collapsed delta function solution. Recently, another article ~\cite{Titsias} was published extending~\cite{Neklyudov2018KLdivergence}. Optimizing a special speed measure ((2) in ~\cite{Titsias}), one can obtain the algorithm which can tune the optimal step-size as well as covariance matrix for the best convergence. Also, it is worth noticing, that entropy is at the heart of Titsias's approach. However, entropy is not the only way to solve such problems. For instance, ~\cite{NICE-MC} and ~\cite{L2HMC} use other methods to get optimal parameters, although the second paper applies their method for HMC. However, we prefer to pick out the Titsias's approach for obtaining optimal step-sizes in our algorithm.

\subsection{Bayesian Linear Regression}
In this subsection, we study the algorithm's ability to converge to a desired distribution arising from Bayesian Linear Regression. We consider a multi-agent network  composed of 4 nodes. Classical linear regression is described by the following expression:
\[ y = Xw + \delta \]
where $y$ constitutes a target variable and the noise is denoted by $\delta$, whereas $X$ and $w$ are a set of features and weights of the linear model, respectively. We can generate data as follows:
\[ \delta \sim \mathcal{N}(0,\sigma^{2}) , \quad X_{j} \space \sim \mathcal{N}(0,I_{1}) \quad \forall j  \]
We generate 800 samples and separate 200 samples across each of four machines. Thus, each agent is going to process its 200 samples, not having access to samples of the other agents.
Our main goal is to get the posterior distribution of weights. However, to facilitate the problem we are going to use "poor" Bayesian inference. In other words, one would like to find out the mode of the posterior distribution. Since we use Bayesian linear regression, we should introduce a prior distribution for weights. Implying features' equal a priori significance, we choose a zero-mean standard normal distribution with  an identity 2 dimensional covariance matrix. Recall Bayes's theorem as follows:
\[ p(w| X,y) = \frac{p(y|X,w)p(w)}{\int_{w}p(y|X,w)p(w)dw}    \]
where $p(w)$ constitutes the prior distribution for weights, while $p(y|X,w)$ is a likelihood of data.  In the experiment, each node will take its own mini-batch, whose size is equal to 1, will calculate consensus step and point where we take the gradient. Having done these computations, each node samples a new object from the posterior distribution of weights. Since our posterior distribution is a multivariate normal distribution and a distribution that makes a current sample on a node is the same, one can calculate the second Wasserstein distance between this distributions analytically as follows:
\[ \mathcal{W}_{2}(\mathcal{N}(m_{1},\Sigma_{1}),\mathcal{N}(m_{2},\Sigma_{2})) = \| m_{1} - m_{2}\| + \|\Sigma_{1}^{\frac{1}{2}} - \Sigma_{2}^{\frac{1}{2}}\|\]
Thus, we repeat this process for 200 iterations (1 batch per one iteration, i.e., one epoch), sometimes changing the set of edges. In figure ~\ref{Bayesian Linear regresssion plots}(a), one can look at the final plot that deals with the convergence between distributions in terms of the second Wasserstein distance. Meanwhile, one can see the consensus error 
for each one from 4 agents in figure ~\ref{Bayesian Linear regresssion plots}(b).

\begin{figure}
\centering
\subfloat[The 2nd Wasserstein distance]{\includegraphics[width = 1.6in]{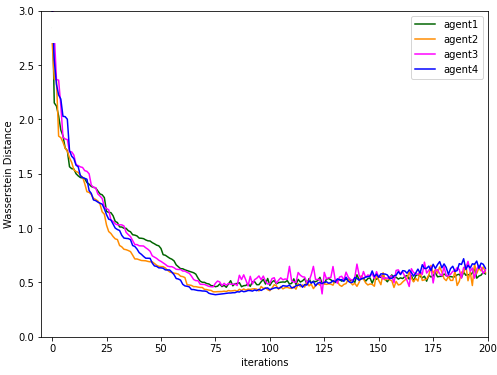}} 
\subfloat[The consensus error]{\includegraphics[width = 1.68in]{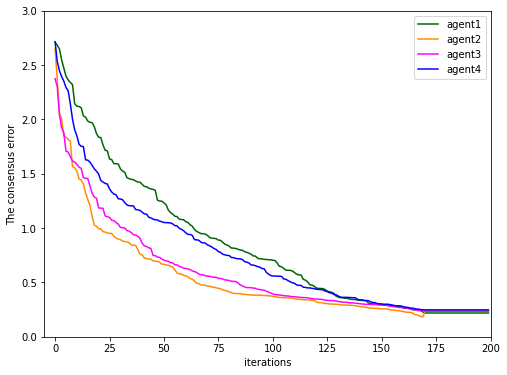}}
\caption{The Second  Wasserstein Distance and the consensus error in case of Bayesian linear regression}
\label{Bayesian Linear regresssion plots}
\end{figure}

\subsection{Sampling from Gaussian mixture}
In this subsection, we illustrate an experiment that deals with  sampling from a  multi-modal probability distribution. Recall the definition of a Gaussian mixture:
\[\sum_{k = 1}^{n} \phi_{k}\mathcal{N}(m_{k},\Sigma_{k})\]
where a the mixture constants $\phi_{k}$ satisfy $\sum_{k=1}^{n}\phi_{k}$ = 1.
Let us consider a mixture of two one-dimensional Gaussian distributions under the assumption that their standard deviations are known, but not their means. However, there are some prior distributions for the means of both Gaussians. Then, we get for each sample $x_{i}$, that:
\begin{equation}
\label{eq: mix_norm}
x_{i} \sim \frac{1}{2}\mathcal{N}(\theta_{1},\sigma_{x}^{2}) +  \frac{1}{2}\mathcal{N}(\theta_{1} + \theta_{2} ,\sigma_{x}^{2})
\end{equation}
where $\sigma_{x}^{2}$ equals to 2. As for prior distributions, they will be as below:
\[\theta_{1} \sim \mathcal{N}(0,\sigma_{1}^{2}),\theta_{2} \sim \mathcal{N}(0,\sigma_{2}^{2}) \]
where $\sigma_{1}^{2}$ and $\sigma_{2}^{2}$
equal to 10 and 1 respectively.
Now, one can fix $\theta_{1}$ = 0 and $\theta_{2}$ = 1. Thus, having fixed values of means, we get the certain mixture of two Gaussian distributions. In this experiment, we consider a decentralized system that is composed of 4 nodes. Then, one can generate 800 samples from the distribution of $x_{i}$ defined above (\ref{eq: mix_norm}) with these given $\theta_1$ and $\theta_2$ and randomly and separate the samples evenly across the network.

Then, the main goal of the experiment is to propose a posterior distribution for $\theta_{1}$ and $\theta_{2}$ from the set of samples, specifically,:
\[  p(\vec{\theta}|X,\sigma_{x},\sigma_{1},\sigma_{2}) \sim [\prod_{i =1}^{n} p(X_{i}|\vec{\theta},\sigma_{x})] p(\vec{\theta}|\sigma_{1},\sigma_{2}) \]
Using MCMC methods and taking the gradients of the joint log-likelihood, one can obtain samples from a desired distribution. We apply the Algorithm given in~\eqref{eq:langdynamicsout}, where we take the gradient of a joint log-likelihood, for each node, for its set of samples, along with the consensus step at each iteration..
The step-size is diminishing and is taken as $\alpha(t) = \frac{a}{(b + t)^{\gamma}}$, where $\gamma$ equals to 0.65, when a and b are set such that $\alpha(t)$ is changed from 0.01 to 0.0001.
In Figure~\ref{Gaussian-mixture-plots}: (a) - (d), one can see how each agent is able to sample from Gaussian mixture posterior distribution above.\\[0.3 cm]
Moreover, one would like to look at an error between each node and average over all nodes on each iteration. Such difference between "average" node and an agent of decentralized network one can call as "the consensus error". In other words, "Tte consensus error" shows up at all the following expression on each iteration for each node in the multi-agent network:
\[ \|X_{(i)}(t) - \bar{X}(t)\|^{2}\]
Thus, one should take values, which are generated by the algorithm from $i$'th node, and calculate the squared difference on each $t$ iteration. One can look at the consensus error for 4 agents in case of sampling from Gaussian mixture in Figure~\ref{Gaussian-mixture-plots}(e). 

\begin{figure}
\centering
\subfloat[agent 1]{\includegraphics[width = 1.68in]{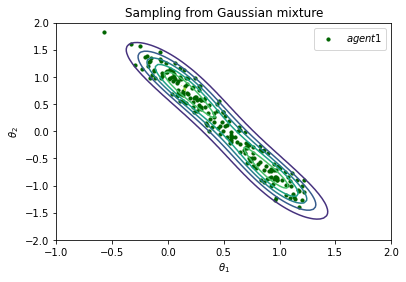}} 
\subfloat[agent 2]{\includegraphics[width = 1.68in]{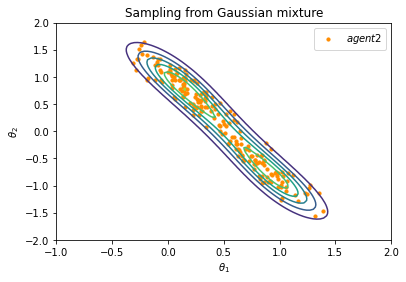}}\\
\subfloat[agent 3]{\includegraphics[width = 1.68in]{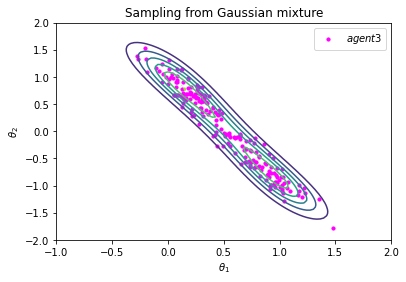}}
\subfloat[agent 4]{\includegraphics[width = 1.68in]{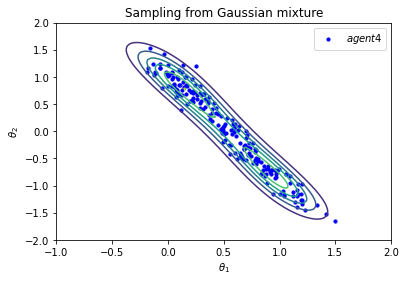}}\\
\subfloat[The consensus error]{\includegraphics[width = 1.68in]{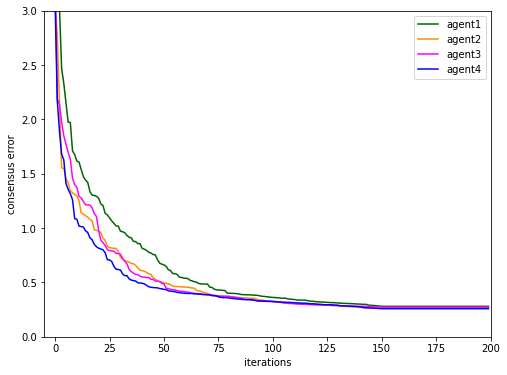}} 
\caption{(a) -(d) :The sampling from Gaussian mixture for 4 agents,
(e) - The consensus error for 4 agents}
\label{Gaussian-mixture-plots}
\end{figure}

We have included results for sampling from a Gaussian mixture in decentralized settings with 8 nodes. Figure~\ref{Bayesian-Logistic-regression-plots} displays the results. One can observe the following expected fact. The more agents takes part in Bayesian logistic regression with decentralized settings, the less variance of accuracy and the faster convergence to the consensus accuracy. The same results are in the experiment for sampling from Gaussian mixture. The more agents are in decentralized settings, the more clearly are seen the two modes of this distribution.

\begin{figure}
\centering
\subfloat[agent 1]{\includegraphics[width = 1.6in]{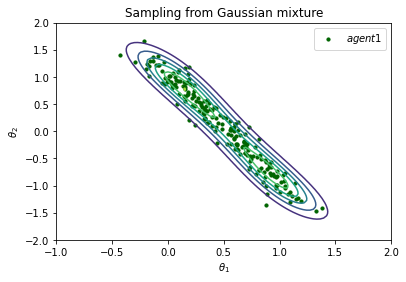}} 
\subfloat[agent 2]{\includegraphics[width = 1.6in]{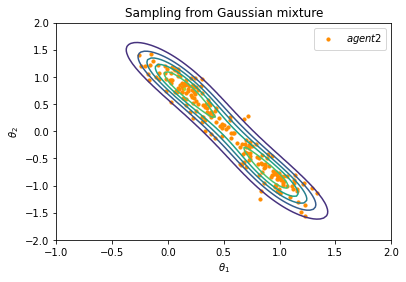}}\\
\subfloat[agent 3]{\includegraphics[width = 1.6in]{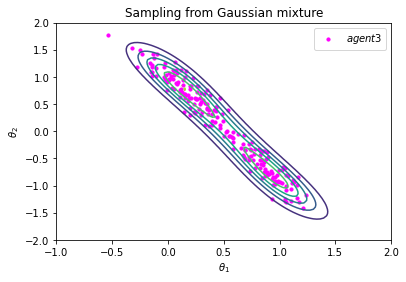}}
\subfloat[agent 4]{\includegraphics[width = 1.6in]{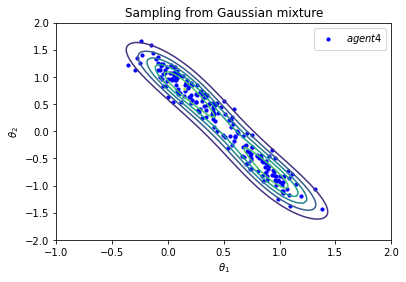}}\\
\subfloat[agent 5]{\includegraphics[width = 1.6in]{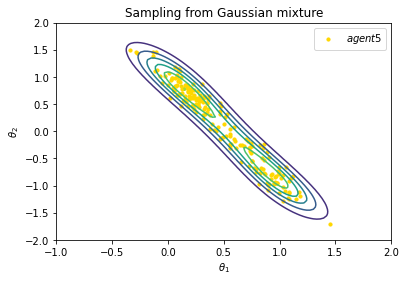}} 
\subfloat[agent 6]{\includegraphics[width = 1.6in]{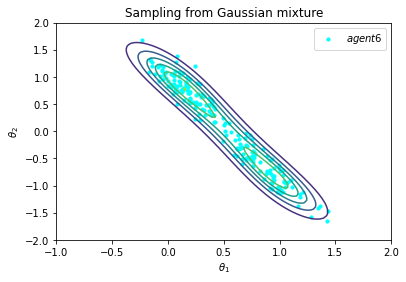}}\\
\caption{Samples from SGLD in case of sampling from Gaussian mixture by 6 agents}
\label{Bayesian SGLD by 6 agents}
\end{figure}

\begin{figure}
\centering
\subfloat[agent 1]{\includegraphics[width = 1.6in]{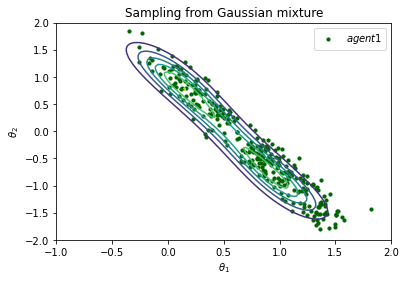}} 
\subfloat[agent 2]{\includegraphics[width = 1.6in]{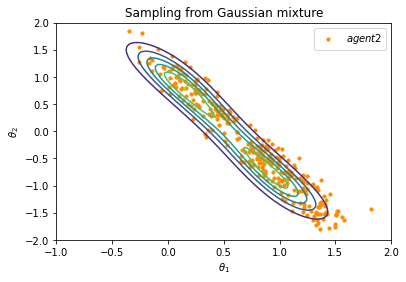}}\\
\subfloat[agent 3]{\includegraphics[width = 1.6in]{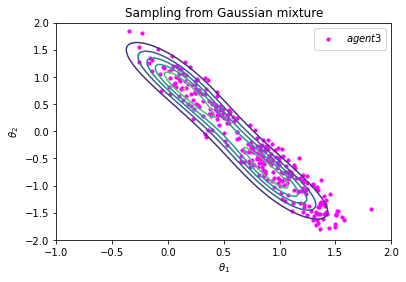}}
\subfloat[agent 4]{\includegraphics[width = 1.6in]{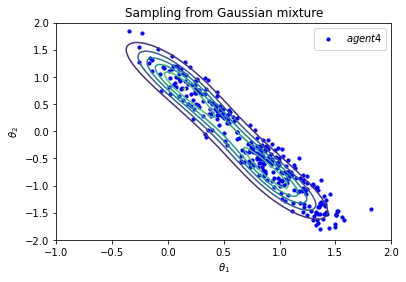}}\\
\subfloat[agent 5]{\includegraphics[width = 1.6in]{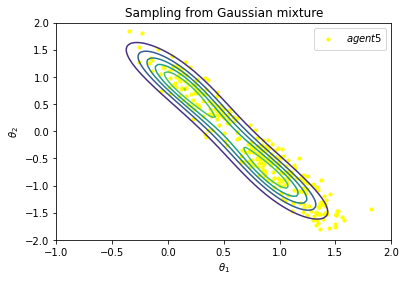}} 
\subfloat[agent 6]{\includegraphics[width = 1.6in]{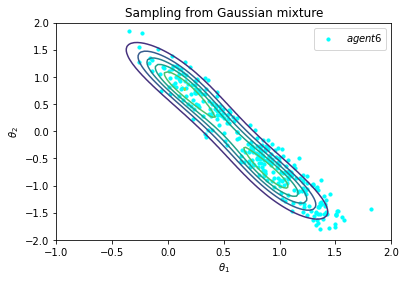}}\\
\subfloat[agent 7]{\includegraphics[width = 1.6in]{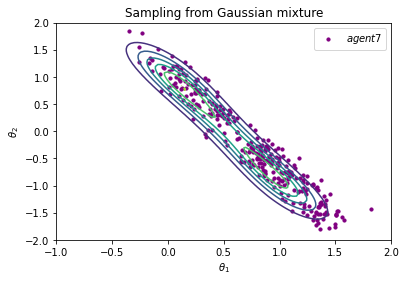}} 
\subfloat[agent 8]{\includegraphics[width = 1.6in]{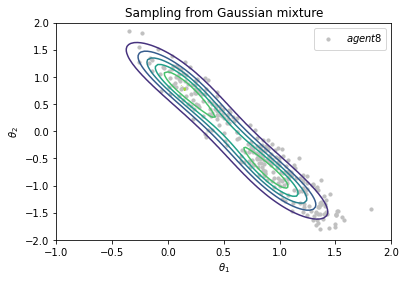}}\\
\subfloat[consensus]{\includegraphics[width = 1.68in]{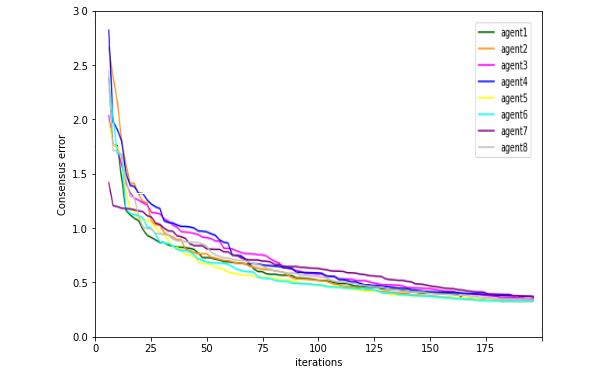}} 
\caption{Samples from SGLD and the Consensus error in case of sampling from Gaussian mixture by 8 agents}
\label{Bayesian-Logistic-regression-plots}
\end{figure}

\begin{figure}
\centering
\subfloat[agent 1]{\includegraphics[width = 1.6in]{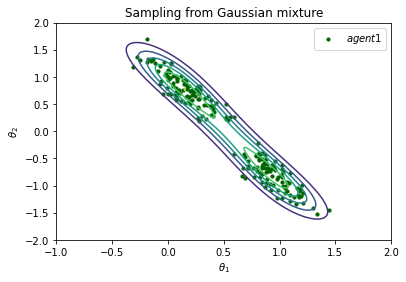}} 
\subfloat[agent 2]{\includegraphics[width = 1.6in]{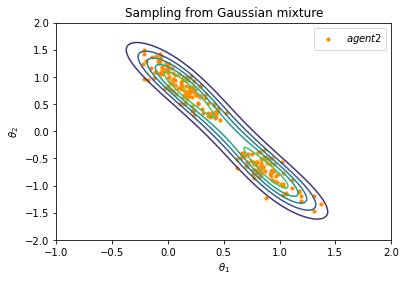}}\\
\subfloat[agent 3]{\includegraphics[width = 1.6in]{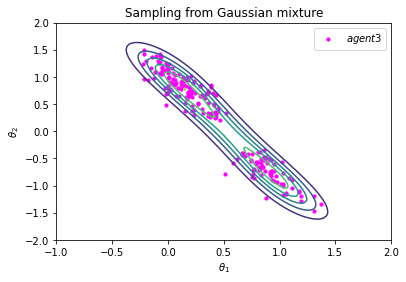}}
\subfloat[agent 4]{\includegraphics[width = 1.6in]{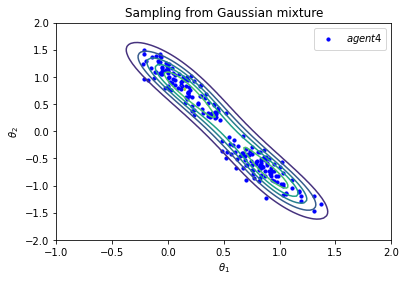}}\\
\subfloat[agent 5]{\includegraphics[width = 1.6in]{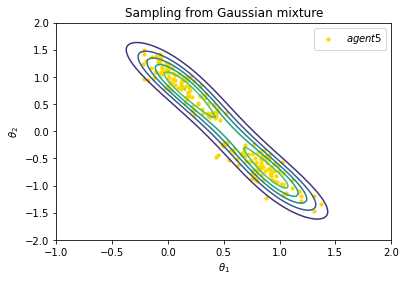}} 
\subfloat[agent 6]{\includegraphics[width = 1.6in]{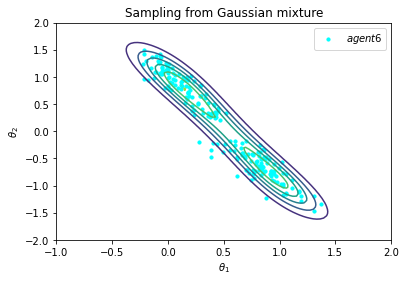}}\\
\subfloat[agent 7]{\includegraphics[width = 1.6in]{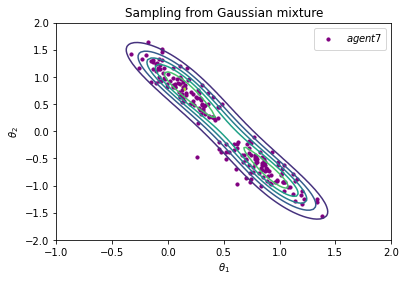}}
\subfloat[agent 8]{\includegraphics[width = 1.6in]{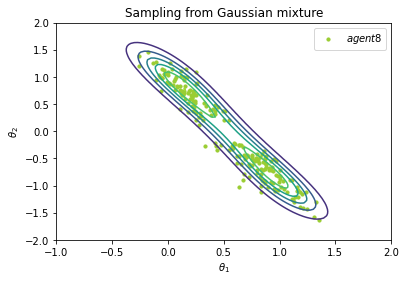}}\\
\subfloat[agent 9]{\includegraphics[width = 1.6in]{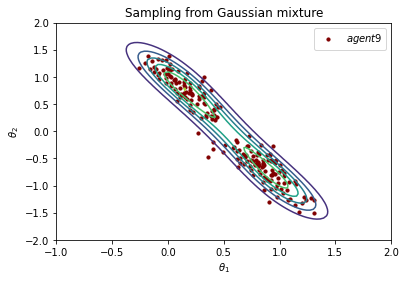}}
\subfloat[agent 10]{\includegraphics[width = 1.6in]{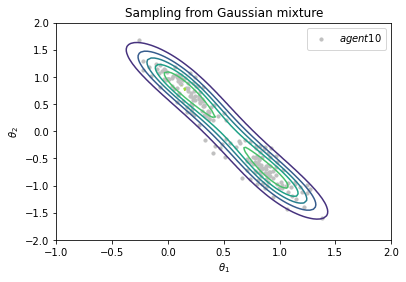}}\\
\caption{Samples from SGLD in case of sampling from Gaussian mixture by 10 agents}
\label{Bayesian SLGD 10 samples}
\end{figure}

\subsection{Bayesian Logistic Regression}
In the final subsection, we apply our algorithm, which is based on stochastic gradient Langevin algorithm, to a case of Bayesian logistic regression for real data. Specifically we consider the problem of binary classification, with a target variable y $\in \{-1, 1\}$. 
The a9a data-test is a huge-scale and classic data-set for binary classification. This data-set is available at the UCI machine learning repository. It is composed of 32561 samples and 123 features. Each feature has 2 unique values 1 or 0, whereas the target variable has values 1 and -1. In this experiment our decentralized setting is composed of 4 agents. Then, we take 80\% of data as train data-set and separate between nodes equally and  randomly. 

We have the following expression for the posterior for $N$ observations:
\[ \pi(w) = p(w|x,y) = \frac{1}{C} \prod_{i = 1}^{N} p(y_{i}|x_{i},w) p(w)\]
where C is a normalizing constant. Then we take~\eqref{eq:langdynamicsout} with $\nabla U_i(Z)=\nabla\left( \prod_{i \in S_i} -\log p(y_{i}|x_{i},Z) p(Z)\right)$ where $S_i\subseteq [N]$ the subset of data given to agent $i$. One can write the expression for the gradient of potential function as follows:
\[  \nabla U(z) = \sum_{i = 1}^{m} \frac{y_{i}x_{i}}{1 + e^{-y_{i}z^{T}x_{i}}} - sign(z) \]
Having computed a consensus solution and weights, where we take the gradient of the potential function for each node, one can recalculate a new value for $w$ and compute the ROC-AUC score on test data-set, whose the size is equal to 20\% of the whole data-set, for each node. The step-size constitutes diminishing step-size and is equal to 
$ \frac{\alpha(0)}{(\gamma + t)^{\phi}}$, where $\alpha(0)$ = 0.008, $\gamma$ = 12, $\phi$ = .45.
In Figure~\ref{Bayesian-Logistic-regression-plotsb}, we plot the ROC-AUC curve for each node in this decentralized setting. We  see the convergence during one epoch (one iteration through the data). The average of nodes reaches  an ROC-AUC score of 84.36\% in 1000 iterations. 

Finally, to see how the accuracy scales with the number of agents, Figure~\ref{Bayesian-Logistic-regression-plots-2} plots the accuracy for each agent for $m\in \{2,6,8,10,14\}$. It can be observe that the cooperative sampling leads to greater stability in the overall performance, and a faster settling of the accuracy curve.

\begin{figure}
\centering
\subfloat[agent 1]{\includegraphics[width = 1.68in]{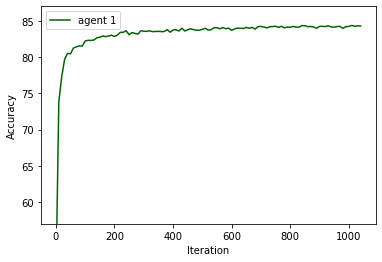}} 
\subfloat[agent 2]{\includegraphics[width = 1.68in]{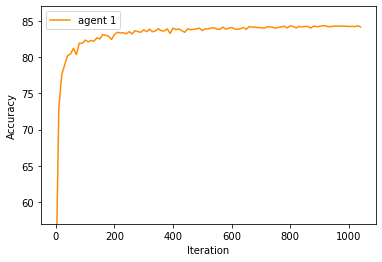}}\\
\subfloat[agent 3]{\includegraphics[width = 1.68in]{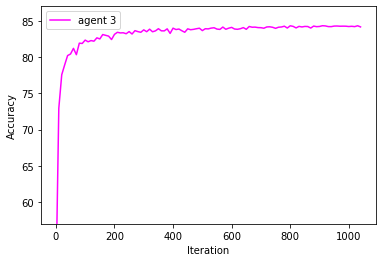}}
\subfloat[agent 4]{\includegraphics[width = 1.68in]{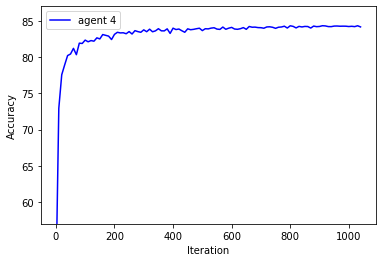}} 
\caption{ROC-AUC curve in Bayesian logistic regression}
\label{Bayesian-Logistic-regression-plotsb}
\end{figure}

\begin{figure}
\centering
\subfloat[2 agents]{\includegraphics[width = 1.6in]{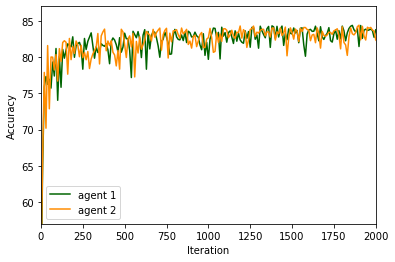}}
\subfloat[6 agents]{\includegraphics[width = 1.6in]{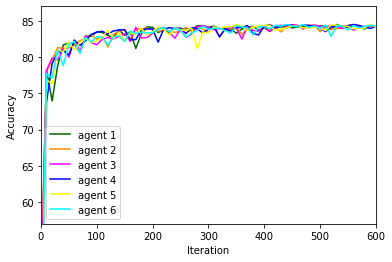}}\\
\subfloat[8 agents]{\includegraphics[width = 1.6in]{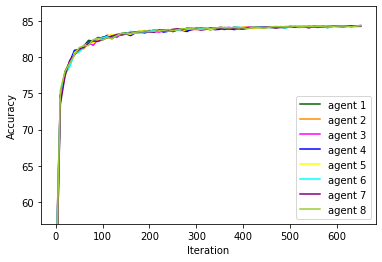}}
\subfloat[10 agents]{\includegraphics[width = 1.6in]{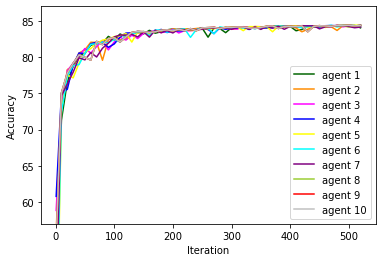}}\\
\subfloat[14 agents]{\includegraphics[width = 1.6in]{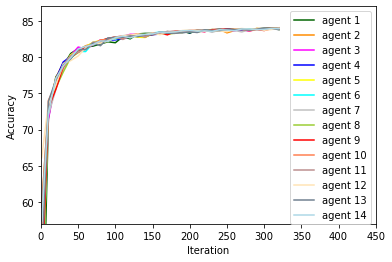}}
\caption{Bayesian Logistic Regression for many agents}
\label{Bayesian-Logistic-regression-plots-2}
\end{figure}

\section{Conclusion}

In this paper we considered the problem of Langevin dynamics for sampling from a distribution with a potential function consisting of data distributed across agents. We consider a decentralized framework wherein the agents communicate by the structure of a directed graph. We proved asymptotic consensus as well as convergence in Wasserstein distance to the stationary distribution of the procedure, and demonstrated the efficacy of the procedure on standard test problems in sampling. 

Intended future work can be studying the properties of the method as a means of finding globally optimal points for nonconvex noisy optimization problems, as well as quantization, study of large scale speedup, and other aspects of sampling in a federated setting.
\section*{Acknowledgements}
The second author would like to acknowledge support from the OP VVV project
CZ.02.1.01/0.0/0.0/16\_019/0000765 Research Center for Informatics

\bibliographystyle{plain}
\bibliography{refs}

\section{Appendix: Proofs of Theoretical Results}

\begin{lemma}\label{lem:helptoboundb}
    It holds that there exists some compact set $\mathcal{W}$ such that for all $i$
    \[
    \mathbb{E}\left\|\frac{X_{(i)}(t+1)}{Y_{(i)}(t+1)}\right\|\le \left\{\begin{array}{lr} \|Z_{(i)}(t+1)\| &
     Z_{(i)}(t+1)\notin\mathcal{W} \\
    R & \text{otherwise}\end{array}\right.
    \]
    with $R$ depending on $\mathcal{W}$ and problem constants.
\end{lemma}
\begin{proof}
This is similar to the proof of Lemma 3 in ~\cite{nedic2016stochastic}.

Notice that,
\[
\begin{array}{l}
X_{(i)}(t+1)/Y_{(i)}(t+1) = Z_{(i)}(t+1)-\alpha(t+1) \nabla U_i(Z_{(i)}(t+1))/Y_{(i)}(t+1) \\ \qquad + 
\sqrt{\alpha(t+1)} B_{(i)}(t+1)/Y_{(i)}(t+1)
\end{array}
\]
Without loss of generality, let $X'_{(i)}(t+1) = X_{(i)}(t+1)/Y_{(i)}(t+1)$. We know that there exists $\delta\le Y_{(i)}(t+1)\le m $ for all $i$ and $t$ as shown in the beginning of the proof of Theorem 2 in~\cite{nedic2016stochastic}.

Indeed we have, for any $Z_{(i)}(t+1)$,
by strong convexity,
\[
\begin{array}{l}
\nabla U_i(Z_{(i)}(t+1))^T Z_{(i)}(t+1)
\ge U_i(Z_{(i)}(t+1))-U_i(0)\\ \qquad\qquad +\frac{\mu}{2} \|Z_{(i)}(t+1)\|^2
\end{array}
\]
Continuing as in the original,
\[
\begin{array}{l}
\|X_{(i)}'(t+1)\|^2 \le (1-\alpha(t+1) \mu / m) 
\|Z_{(i)}(t+1)\|^2 \\ \qquad-2\alpha(t+1) (U_i(Z_{(i)}(t+1))-U_i(0))/Y_{(i)}(t+1)^2 \\ \qquad  - 2\sqrt{\alpha(t+1)}
B_{(i)}(t+1)^T Z_{(i)}(t+1)/Y_{(i)}(t+1)^2
\\ \qquad + 2\alpha(t+1)^2\|\nabla U_i(Z_{(i)}(t+1))\|^2/Y_{(i)}(t+1)\\ \qquad +4\alpha(t+1) \|B_{(i)}(t+1)\|^2/Y_{(i)}(t+1)^2
\end{array}
\]
We have that,
\[
\|\nabla U_i(Z_{(i)}(t+1))\|^2 \le 
2L^2 \|Z_{(i)}(t+1)\|^2+2 \|\nabla U_i(0)\|^2
\]
and so,
\[
\begin{array}{l}
\|X_{(i)}'(t+1)\|^2 \\ \qquad \le \left(1-\alpha(t+1)\mu/m +2\alpha(t+1) L^2 /\delta^2\right) \|Z_{(i)}(t+1)\|^2 \\ \qquad \qquad -2\alpha(t+1) (U_i(Z_{(i)}(t+1))-U_i(0))/m \\ \qquad \qquad - 2\sqrt{\alpha(t+1)}
B_{(i)}(t+1)^T Z_{(i)}(t+1)/m
\\ \qquad \qquad+ 4\alpha(t+1)^2/\delta^2 \|\nabla U_i(0)\|^2\\ \qquad \qquad +4\alpha(t+1) /\delta^2 \|B_{(i)}(t+1)\|^2
\end{array}
\]
Taking conditional expectations on the filtration,
\begin{equation}\label{eq:ztox}
\begin{array}{l}
\mathbb{E} \|X_{(i)}'(t+1)\|^2 \\ \qquad \le \left(1-\alpha(t+1)\mu/m +2\alpha(t+1) L^2/\delta^2\right)  \|Z_{(i)}(t+1)\|^2 \\ \qquad \qquad -2\alpha(t+1)/m (U_i(Z_{(i)}(t+1))-U_i(0)) \\ \qquad \qquad+ 4\alpha(t+1)^2/\delta^2\|\nabla U_i(0)\|^2\\ \qquad \qquad +4\alpha(t+1)\sigma^2/\delta^2
\end{array}
\end{equation}
Enforce $\alpha(t+1) \le \frac{\mu}{4L^2}$. 
We can define the set $\mathcal{V}$ to be such that,
\[
\mathcal{V}:=\left\{z | U_i(z) \le U_i(0)+\frac{2\mu}{L^2} \|\nabla U_i(0)\|^2 +2\sigma^2\right\}
\]
for all $i$.
Since $U_i$ is strongly convex, this set is compact. Now clearly if $Z_{(i)}\notin \mathcal{V}$ then the last term in~\eqref{eq:ztox} is negative and so
by the requirement of $\alpha(t)$ the first term has a coefficient less than one, and so we have
that $\mathbb{E} \|X_{(i)}'(t+1)\|\le \|Z_{(i)}(t+1)\|$. 

Otherwise (i.e., $Z_{(i)}(t+1)\in\mathcal{W}$), we have that,
\[
\begin{array}{l}
\mathbb{E}\|X_{(i)}'(t+1)\| \le  \|Z_{(i)}(t+1)\|+\alpha(t+1) \|U_i(Z_{(i)}(t+1))\|/\delta \\ \qquad \qquad\qquad +\sqrt{\alpha(t+1)}\sigma/\delta
\end{array}
\]
and the statement follows from the compactness
of $\mathcal{W}$. 

\end{proof}

\begin{lemma}
    It holds that there exists a $C$ such that,
    \[
    \mathbb{E}\|\nabla U_i(Z_{(i)}(t+1)\|\le C
    \]
\end{lemma}

\begin{proof}
By the same argument as for Theorem 2 in~\cite{nedic2016stochastic}, we can conclude that $\mathbb{E}\|Z_{(i)}(t+1)\|\leq D < \infty$. Indeed, 
we have by the same argument that $1\le Y_{(i)}(t)\le m$, and
\[
\max_i \mathbb{E}\|Z_{(i)}(t+1)\|\leq \max_j\mathbb{E}\left\|\frac{X_{(j)}(t)}{Y_{(j)}(t)}\right\| \le \max \left\{\max_i\|Z_{(i)}(t)\|, R\right\}
\]
and the statement following from induction and iterated expectations. 

Since the gradient of $U(Z_{(i)}(t+1))$ is L-Lipschitz, we have,
\[ \Vert \nabla U(Z_{(i)(t+1)}) - \nabla U(X^{*}) \Vert \leq L \Vert Z_{(i)}(t+1) - X^{*} \Vert  \]
where $X^*$ is the unique global minimizer of $U(x)$, which exists by strong convexity. In accordance to the property of norms:
\[ \Vert a \Vert -\Vert b \Vert \leq \Vert a-b \Vert\]
we get,
\[ \Vert \nabla U(Z_{(i)}(t+1)) \Vert - \Vert \nabla U(X^{*}) \Vert  \leq L \Vert Z_{(i)}(t+1) - X^{*} \Vert  \]
and
\[  \Vert \nabla U(Z_{(i)}(t+1)) \Vert \leq L  \Vert Z_{(i)}(t+1) - X^{*} \Vert + \Vert \nabla U_{(i)}(X^{*}) \Vert \]
Hence:
\[ \mathbb{E} \Vert \nabla U_{(i)}(Z_{(i)}(t+1))\Vert \leq L \mathbb{E}\Vert Z_{(i)}(t+1) - X^{*} \Vert + \Vert \nabla U_{(i)}(X^{*}) \Vert \]
Since the gradient of $U(X^{*})$ equals zero, then one can show
$\mathbb{E}\Vert Z_{(i)}(t+1) - X^{*} \Vert$ is bounded, and the reverse triangle inequality implies the final result.

\end{proof}

This is the same as Corollary 1 in ~\cite{nedic2016stochastic},
\begin{lemma}
    It holds that,
    \[
    \begin{array}{l}
    \left\|Z_{(i)}(t+1)-\frac{\sum\limits_{i=1}^m X_{(i)}(t)}{m}\right\|\\ \qquad \le 
    \frac{8}{\delta}\left(\lambda^t\sum\limits_{i=1}^m\|X_{(i)}(0)\|_1+
    \sum\limits_{s=1}^t \lambda^{t-s}\sum\limits_{i=1}^m
    \left\| e_i(s)\right\|\right)
    \end{array}
    \]
    where,
    \[
    e_i(s) = \alpha(s+1) \nabla U_{(i)}(Z_{(i)}(s + 1)))  + \sqrt{2 \alpha(s+1)}B(s+1)_{(i)}
    \]
\end{lemma}
\begin{proof}
The proof of the lemma is a direct analog of the proof Lemma 1 in ~\cite{NedicDirected2014}.
\end{proof}

\begin{lemma}\label{lem:conserrorb}
\[
\begin{array}{l}
    \mathbb{E}\left[\sum\limits_{t=1}^\tau \left\|Z_{(i)}(t+1)-\frac{\sum\limits_{i=1}^m X_{(i)}(t)}{m}\right\| \right] \\
    \qquad \le \frac{8}{\delta}\frac{\lambda}{1-\lambda}\sum\limits_{i=1}^m 
    \|X_{(i)}(0)\|_1+\frac{8}{\delta}\frac{Dm}{1-\lambda} (1+\sqrt{\tau})
\end{array}
\]
\end{lemma}
\begin{proof}
The proof of this result is the same as Corollary 2 in ~\cite{nedic2016stochastic}. Note that because of the square root of $\alpha$ term in the noise, we integrate out $1/\sqrt{t}$ instead of $1/t$, hence the $\sqrt{\tau}$ instead of $\log \tau$ here.
\end{proof}

Now define,
\[
\bar{X}(t) = \frac{\sum\limits_{i=1}^m X_{(i)}(t)}{m}
\]
\begin{lemma}
We have, by the column-stochasticity of 
$A(t)$,
\begin{equation}\label{eq:avgupdateb}
\begin{array}{l}
\bar{X}(t+1)=\bar{X}(t) - \frac{\alpha(t+1)}{m}\sum\limits_{i=1}^m \nabla U(\bar{X}(t)) \\ \qquad +  \frac{\alpha(t+1)}{m}\sum\limits_{i=1}^m \left(\nabla U(\bar{X}(t))-U(Z_{(i)}(t+1))\right)+\bar{B}(t+1)
\end{array}
\end{equation}
\end{lemma}
\begin{proof}
In accordance with the algorithm one can rewrite the final equation system as follows:
\[X_{(i)}(t+1) = \vert A(t)X(t)\vert_{(i)} - \alpha(t+1)\nabla U_{(i)}(Z_{(i)}(t+1)) + B_{(i)}(t+1)\]
Then, one can average a column-vector that is composed of such $X_{(i)}$, where $i \in \{1,..,m\}$ . Then:
\[ 
\begin{array}{l}\overline{X}(t+1) = \sum_{i=1}^{m}\left\vert \frac{A(t)X(t)}{m}\right\vert_{(i)} - \frac{\alpha(t+1)}{m}\sum_{i=1}^{m}\nabla U_{(i)}(Z_{(i)}(t+1)) \\ \qquad\qquad\qquad + \overline{B}(t+1)\end{array}\]
Having recalled A is a column-stochastic matrix, then :
\[1^{T}u = 1^{T}A(t)u\]
for any vector $u$. Then:
\[ \frac{1}{m} \sum_{i=1}^{m} \vert 1 A(t)X(t)\vert_{(i)} = \frac{1}{m} \sum_{i=1}^{m} X_{(i)}=  \overline{X}(t) \]
Then, one can rewrite the scheme as follow:
\[ \overline{X}(t+1) = \overline{X}(t) - \frac{\alpha(t+1)}{m}\sum_{i=1}^{m}\nabla U_{(i)}(Z_{(i)}(t+1)) + \overline{B}(t+1)\]
Adding and subtracting $\frac{\alpha(t+1)}{m}\sum_{i=1}^{m} \nabla U(\overline{X}(t))$, we get the desired expression.
\end{proof}
We have the following Corollary,
\begin{corollary}\label{cor:consb}
Given any $\gamma>0$, there exists $C$ positive such that,
\[
E\|\bar{X}(t))-Z_{(i)}(t+1)\| \le \frac{C}{t^{1/2-\gamma}}
\]
\end{corollary}

We are now ready to use the arguments in ~\cite{dalalyan2019user} regarding a perturbed Langevin methods. In particular, we can apply Proposition 2 to state that,
\begin{lemma}
\[
\begin{array}{l}
W_2(\bar{\nu}_{t+1},\pi) \le \rho_{t+1}W_2(\nu_t,\pi)+1.65 L (\alpha_{t+1}^3 d)^{1/2}\\ \qquad +\alpha(t+1)\sqrt{d} L \sum\limits_{i=1}^m \mathbb E\|\bar{X}(t))-Z_{(i)}(t+1)\|
\end{array}
\]
where $\rho_{t+1} = \max(1-\mu \alpha(t+1),L\alpha(t+1)-1)$.
\end{lemma}

Now recall
\begin{lemma}\cite[Lemma 2.4]{polyak1987introduction}\label{lem:polyakb}
Let $u_k\ge 0$ and,
\[
u_{k+1}\le \left(1-\frac{c}{k}\right) u_k+\frac{d}{k^{p+1}}
\]
with $d>0$, $p>0$ and $c>0$ and $c>p$. Then,
\[
u_k\le d(c-p)^{-1}k^{-p}+o(k^{-p})
\]
\end{lemma}

Apply the previous two Lemmas together with Corollary~\ref{cor:cons} to conclude,
\begin{theorem}\label{th:convb}
Let $\alpha(0)\le \min\left\{\frac{1}{2L},\frac{\mu}{4L^2}\right\}$ and
$\alpha(t)=\alpha(0)/(1+t)$. We have that,
\[
W_2(\bar{\nu}_{k+1},\pi) \le \frac{(C_\gamma m+1.65) Ld^{1/2}}{(\mu\alpha(0)-1/2+\gamma) (1+t)^{(1/2-\gamma)}}+\beta_k
\]
where $\beta_k=o\left((1+t)^{-(1/2-\gamma)}\right)$
\end{theorem}

\end{document}